\documentclass[conference]{IEEEtran}
\usepackage{cite}
\usepackage{xcolor}

\usepackage{graphicx,amssymb,amstext,amsmath}

\newcounter{actr}
{\begin{list}{(\alph{actr})}{\usecounter{actr}}}{\end{list}}

\newcounter{ictr}
{\begin{list}{(\roman{ictr})}{\usecounter{ictr}}}{\end{list}}

\newtheorem{remark}{Remark}
\newtheorem{thm}{Theorem}

\newtheorem{prop}{Proposition}
\newtheorem{defn}{Definition}
\newtheorem{fact}{Fact}
\newenvironment{new-proof}[1]
{{\em Proof }:\\}%
{ \noindent\qed }
\newcommand{\mrm}{\mathrm}

\newcommand{\cE}{{\mathcal{E}}}

\newcommand{\bG}{{\mathbf{G}}}

\newcommand{\bH}{{\mathbf{H}}}

\newcommand{\bp}{{\mathbf{p}}}

\newcommand{\bs}{{\mathbf{s}}}

\newcommand{\cS}{{\mathcal{S}}}

\newcommand{\bu}{{\mathbf{u}}}

\newcommand{\bv}{{\mathbf{v}}}

\newcommand{\cW}{{\mathcal{W}}}

\newcommand{\bx}{{\mathbf{x}}}

\newcommand{\cX}{{\mathcal{X}}}

\newcommand{\al}{\alpha}

\newcommand{\eps}{\varepsilon}

\DeclareMathAlphabet{\mathbsf}{OT1}{cmss}{bx}{n}
\DeclareMathAlphabet{\mathssf}{OT1}{cmss}{m}{sl}

\DeclareSymbolFont{bsfletters}{OT1}{cmss}{bx}{n}
\DeclareSymbolFont{ssfletters}{OT1}{cmss}{m}{n}
\DeclareMathSymbol{\bsfGamma}{0}{bsfletters}{'000}
\DeclareMathSymbol{\ssfGamma}{0}{ssfletters}{'000}
\DeclareMathSymbol{\bsfDelta}{0}{bsfletters}{'001}
\DeclareMathSymbol{\ssfDelta}{0}{ssfletters}{'001}
\DeclareMathSymbol{\bsfTheta}{0}{bsfletters}{'002}
\DeclareMathSymbol{\ssfTheta}{0}{ssfletters}{'002}
\DeclareMathSymbol{\bsfLambda}{0}{bsfletters}{'003}
\DeclareMathSymbol{\ssfLambda}{0}{ssfletters}{'003}
\DeclareMathSymbol{\bsfXi}{0}{bsfletters}{'004}
\DeclareMathSymbol{\ssfXi}{0}{ssfletters}{'004}
\DeclareMathSymbol{\bsfPi}{0}{bsfletters}{'005}
\DeclareMathSymbol{\ssfPi}{0}{ssfletters}{'005}
\DeclareMathSymbol{\bsfSigma}{0}{bsfletters}{'006}
\DeclareMathSymbol{\ssfSigma}{0}{ssfletters}{'006}
\DeclareMathSymbol{\bsfUpsilon}{0}{bsfletters}{'007}
\DeclareMathSymbol{\ssfUpsilon}{0}{ssfletters}{'007}
\DeclareMathSymbol{\bsfPhi}{0}{bsfletters}{'010}
\DeclareMathSymbol{\ssfPhi}{0}{ssfletters}{'010}
\DeclareMathSymbol{\bsfPsi}{0}{bsfletters}{'011}
\DeclareMathSymbol{\ssfPsi}{0}{ssfletters}{'011}
\DeclareMathSymbol{\bsfOmega}{0}{bsfletters}{'012}
\DeclareMathSymbol{\ssfOmega}{0}{ssfletters}{'012}

\newcommand{\rvs}{{\mathssf{s}}}    

\newcommand{\rvu}{{\mathssf{u}}}    

\newcommand{\rvv}{{\mathssf{v}}}    

\newcommand{\rvx}{{\mathssf{x}}}    

\newcommand{\rvy}{{\mathssf{y}}}    

\ifCLASSINFOpdf
\else
\fi
\newcommand{\bracedn}[4]{\draw[decorate, decoration={brace, amplitude=5pt},thick] ([xshift=-0.5mm,yshift=#3]#2.south east)--([xshift=0.5mm,yshift=#3]#1.south west) node[midway,anchor=north,outer sep=2mm] {#4}}
\newcommand{\dimup}[4]{\draw[<->,thick] ([xshift=0.5mm,yshift=#3]#1.north west)--([xshift=-0.5mm,yshift=#3]#2.north east) node[midway,anchor=south] {#4}}
\newcommand{\dimdn}[4]{\draw[<->,thick] ([xshift=-0.5mm,yshift=#3]#2.south east)--([xshift=0.5mm,yshift=#3]#1.south west) node[midway,anchor=north] {#4}}
\definecolor{light-gray}{gray}{0.75}

\usepackage{tikz}
\usetikzlibrary{shapes,arrows,shadows,positioning,decorations.pathreplacing,trees,calc}
\tikzstyle{sym} = [draw, thick, rectangle, font=\small, minimum width=4mm, minimum height=4mm, text centered]
\tikzstyle{esym} = [sym, fill=light-gray]
\tikzstyle{usym} = [sym, fill=white]
\tikzstyle{diagbox} = [draw, rectangle, font=\footnotesize, fill=white, text centered, rounded corners]
\tikzstyle{codebox} = [draw, rectangle, font=\footnotesize, minimum height=7mm, fill=white, text centered]

\begin{document}
%
\title{Streaming Codes for Channels with Burst and Isolated Erasures}

\author{\IEEEauthorblockN{Ahmed Badr and Ashish Khisti}
\IEEEauthorblockA{ Electrical and Computer Engineering\\
University of Toronto\\
Toronto, ON, M5S 3G4\\
Email: \{abadr, akhisti\}@comm.utoronto.ca}
\and
\IEEEauthorblockN{Wai-Tian Tan and John Apostolopoulos}
\IEEEauthorblockA{Mobile and Immersive Experience Lab\\
Hewlett Packard Laboratories\\
1501 Pagemill Road\\
Palo Alto, CA, 94304}
}


%


\maketitle

\begin{abstract}
We study low-delay error correction codes for streaming recovery over a class of packet-erasure channels that introduce both burst-erasures and isolated erasures. We propose a simple, yet effective class of codes whose parameters can be tuned to obtain a tradeoff between the capability to correct burst and isolated erasures. Our construction generalizes previously proposed low-delay codes which are effective only against burst erasures.  

We establish an information theoretic upper bound on the capability of any code to simultaneously correct burst and isolated erasures and show that our proposed constructions meet the upper bound in some special cases. We discuss the operational significance of column-distance and column-span metrics and establish that the rate $1/2$ codes discovered by Martinian and Sundberg [IT Trans.\, 2004] through a computer search indeed attain the optimal column-distance and column-span tradeoff.

Numerical simulations over a Gilbert-Elliott channel model and a Fritchman model show significant performance gains  over previously proposed low-delay codes and random linear codes for certain range of channel parameters.
\end{abstract}


%
\IEEEpeerreviewmaketitle

\section{Introduction}
Emerging applications such as interactive video conferencing, voice over IP and cloud computing
are required to achieve an  end-to-end latency of less than 200 ms. The
round-trip time in traditional networks can alone approach this limit. Hence it is necessary 
to develop new delay-optimized networking protocols and delay-sensitive
coding techniques in order to meet such stringent delay constraints. 
In this paper we focus on  low-delay error correction codes 
for streaming data at the application layer. 
Commonly used error correction codes operate on message blocks. To apply them to streaming data,
we need to either buffer data packets at the encoder or accumulate all packets at the decoder before
any recovery is possible. To reduce delay we need to keep the codeword lengths short, which in turn reduces the error correction capability.

The fundamental limits of delay-constrained communication are very different from the classical
Shannon capacity.  It is well known for example that the Shannon capacity of an erasure channel only
depends on the fraction of the packets lost over the channel. However when delay constraints
are imposed, the pattern of packet losses becomes significant. As a toy example, consider two different
communication channels as shown in Fig.~\ref{fig:chan-ex} with different loss patterns. The first channel introduces up-to two erasures
in any sliding window of length four. The second channel can erase up-to four packets in a burst, 
but any burst must be followed by a guard interval of at-least four non-erased packets.  Clearly both channel models
have a loss rate of $50\%$. However the decoding deadlines that
can be realized over these channels can be very different. For the first channel, we can use a short $(4,2)$
erasure-correction code and recover each source packet with a deadline of $\tau =4$ time units. For the second channel
we  need to use a $(8,4)$ erasure correction code  and this yields a deadline of $\tau=8$ time units.

\begin{figure}
\centering
\includegraphics[width=\linewidth]{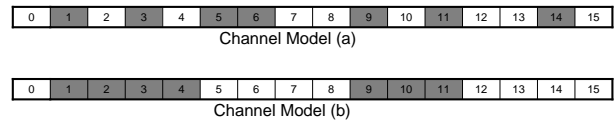}
\caption{Two packet erasure channels with a different loss structure. The first channel has no more than two 
erasures in a sliding window of length four whereas the second channel can have up-to four erasures in a single burst
followed by a guard spacing of at-least four non-erased packets. The shaded packets are erased symbols. A similar example also appears in~\cite{martinianThesis}.}
\label{fig:chan-ex}
\end{figure}

Surprisingly it turns out that the decoding delay on the second channel can  be reduced to ${\tau=5}$ by using a rate $1/2$  delay-optimal code for the burst-erasure channel  proposed in~\cite{MartinianT07, MartinianS04,martinianThesis}. Unlike traditional codes, these constructions recognize the different recovery deadlines of streaming data, and do not wait to recover all the erased packets simultaneously. Instead they exploit the burst-structure of the channel to enable selective recovery of earlier data. In particular, following the erasure burst between $t \in [1,4]$ the code recovers only the data packet $s[1]$ at time $t=5$, the data packet $s[2]$ at time $t=6$ etc.  Such low-delay constructions exist for any burst-erasure channel with a maximum burst-length and a given delay. We will refer to these constructions as streaming codes (SCo) in this paper and the associated feature of recovering successive source packets in a sequential manner as streaming recovery. 

One weakness of the SCo codes~\cite{MartinianT07, MartinianS04,martinianThesis} is that their performance is sensitive to isolated packet losses. As  reported in our simulations  over a Gilbert-Eliott channel model, the error-correction capability of the code deteriorates significantly when we introduce just a small loss probability in the good state.  Motivated by this observation,  we study low-delay error correction codes for a class of channels  that introduce both burst erasures and isolated erasures. Fig.~\ref{fig:chan-mix} provides an example of such a channel. In any sliding window of a given length $W$, the channel can introduce either a certain number of erasures in arbitrary locations or an erasure burst of a certain maximum length. As we observe in simulations, low-delay codes for such channels also perform well over Gilbert Eliott channels and other related channels.

One simple construction for such channels is based on concatenation of two different codes. 
We  generate one set of parity checks from a standard erasure code and another set from the SCo code and then concatenate the two parity checks in the transmitted packet.
The former parity checks can be used when the window of interest has isolated erasures whereas the latter parity checks can be used when it has burst-erasures. Unfortunately such an approach can introduce a significant overhead and is not desirable.

From a code design viewpoint, codes with large {column distance} can correct large number of isolated erasures, whereas codes with large {column span} can correct large bursts. Thus we seek codes with large {column distance} $(d_T)$ and {column span} $(c_T)$ for channels with both burst and isolated losses.
Naturally there exists a tradeoff between these parameters. We establish, to our knowledge, the first information theoretic
outer bound on the achievable  $(d_T,c_T)$ for any code of a given rate. This bound enables us to verify that some of
the code constructions reported using a computer search in~\cite{MartinianS04} are indeed optimal. 

Our proposed construction divides each source packet $s[i]$ into two groups of sub packets say $s_A[i]$ and $s_B[i]$. 
It generates separate parity checks $p_A[\cdot]$ and $p_B[\cdot]$ for each group and  combines the parity checks $p_A[t] + p_B[t-\Delta]$  
after a suitable time-shift of $\Delta$.  By increasing the shift $\Delta$ we tradeoff the column distance for a larger column span. 
Our  construction is optimal for $R=1/2$. Codes with either a maximum value of $d_T$ or $c_T$  also appear as special cases in this construction. 


One practical appeal of our constructions is the ability to perform trade-off between correcting burst and isolated losses using a simple mechanism. This means the same encoder and decoder can work with different channels with different mix of burst and isolated losses by simply adjusting the shift $\Delta$. Furthermore, such trade-off can be adjusted mid-session if the application identifies a change in prevalent network conditions. Since only a single parameter is involved, it contains negligible overhead to send $\Delta$ in each packet so that trade-offs can be made without explicit signalling that could add delay.

We point the reader to~\cite{de-sco,mu-sco,liKG:11, LuiMASc, tekin,leong,streaming-1, streaming-9,streaming-2, streaming-5,streaming-6,streaming-7,streaming-10,streaming-11} for additional works on error control mechanisms for streaming. 

\section{System Model}
We study low-delay error correction codes for a particular channel model  with the following property.
Take any sliding window of length $W$. The channel can introduce either a single erasure burst of length $B$ or a maximum of $N$ erasures in
arbitrary locations, but no other erasure pattern. We will generally assume that ${N < B}$ since the set of arbitrary erasures includes the burst-erasure
pattern as a special case.  An example of such a channel with $W=5$, $B=3$ and $N=2$ is provided in Fig.~\ref{fig:chan-mix}.

\begin{figure}
\centering
\includegraphics[width=\linewidth]{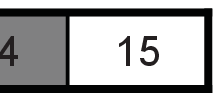}
\caption{A channel model with a mixture of burst-erasures and isolated erasures. In any sliding window of length $W=5$ there is either a single erasure burst of length $B=3$ or
up-to $N=2$ erasures.}
\label{fig:chan-mix}
\vspace{-1em}
\end{figure}

We assume a deterministic source arrival process. At time $i \ge 0,$ the encoder is revealed a source packet $\rvs[i]$ which we assume is a symbol from a source alphabet $\cS$.
At time $i$ the encoder generates a channel symbol $\rvx[i]$ which belongs to a channel input alphabet $\cX$.  The channel symbol is a causal function of the source symbols, i.e.
\begin{align}
\rvx[i] = f_i(\rvs[0], \ldots, \rvs[i]), \quad i \ge 0. \label{eq:enc-func}
\end{align}

The channel output is given by either $\rvy[i] =\rvx[i],$ when the packet is not erased and by $\rvy[i] = \star,$ when the packet is erased. Given the channel output, the decoder is required
to reconstruct each packet with a delay of $T$ units i.e.,\footnote{Notice that the total number of channel packets involving $\rvs[i]$ before its recovery is $T+1$. }
\begin{align}
\rvs[i] = g_i(\rvy[0], \ldots, \rvy[i+T]).\label{eq:dec-func}
\end{align}


\begin{remark}
In contrast to $(n,k)$ block code, where $k$ information symbols are mapped to $n$ codeword symbols, the proposed setup maps a stream of incoming source packets over
an alphabet $\cS$ to a stream of channel packets over the alphabet $\cX$. To add redundancy we require that $|\cX| \ge |\cS|$. 
\end{remark}


A rate $R = \frac{|\cS|}{|\cX|}$ is achievable if there exists a feasible code that recovers every erased symbol $s[i]$ by time ${i+T}$ from any permissible channel 
i.e., the channel introduces no more than $N$ arbitrary erasures or a single erasure-burst of length up-to $B$ in any sliding window of length $W$.

For the rest of the paper, we set $W=T+1$ as the analysis is most convenient for this choice. The interplay between delay and the channel-dynamics also appears most interesting in this regime. For $T \gg W$ the delay constraint is not particularly active, while for $T \ll W$ the guard separation between packet losses can be generally large. 

\section{Distance and Span Metrics}
\label{sec:dist-col}

Let ${\mathbb F}_q$ denote a finite-field of size $q$. For convenience we let $\cS = {\mathbb{F}}_q^k$ and $\cX = \mathbb{F}_q^n$. We view the input symbols $\rvs[i] \equiv \bs_i$ as a length $k$ vector  over ${\mathbb F}_q$ and $\rvx[i] \equiv \bx_i$ as a length $n$ vector over ${\mathbb F}_q$. 
We  restrict our attention to time-invariant linear $(n,k,m)$ convolutional codes specified by
$\bx_i = \sum_{j=0}^m \bs_{i-j} \bG_{j}$
where $\bG_0,\ldots, \bG_m$ are  generator matrices over ${\mathbb F}_q^{k \times n}$. 

The first $T$ output symbols can be expressed as,
\begin{align}
\vspace{-1em}
\label{eq:trunc-cc}
[\bx_0, \bx_1, \ldots, \bx_T] = [\bs_0, \bs_1, \ldots, \bs_T] \cdot \bG^s_T.
\end{align}
where 
\begin{equation}\bG^s_T = \begin{bmatrix}\bG_0 & \bG_1 & \ldots & \bG_T \\  0 & \bG_0 &  & \bG_{T-1} \\ \vdots & &\ddots & \vdots \\ 0 & & \ldots & \bG_0 \end{bmatrix}\label{eq:GsT}\end{equation}
is the truncated generator matrix to the first ${T+1}$ columns. Note that $\bG_j =0$ if $j > m$. For the low-delay property the
 minimum distance and span properties of $\bG_T^s$ are important as discussed below. Such a connection was  
 discussed in~\cite{MartinianS04} and used to perform a computer search of good low-delay codes. 

\begin{defn}[Column Distance]
The column distance of $\bG_T^s$ is defined as
$$d_T  = \min_{\substack{\bs \equiv [\bs_0, \bs_1, \ldots, \bs_T]\\ \bs_0 \neq 0}} \mathrm{wt}(\bs \cdot \bG^s_T)$$
where $\mrm{wt}(\bx)$ equals to the Hamming weight of the vector $\bx$. 
\end{defn}

We refer the reader to~\cite[Chapter 3]{zigangirov} for some properties of $d_T$. 
\begin{fact}
A  convolutional code with a column distance of $d_T$ can recover every information symbol with a delay of $T$ provided the channel introduces no more than
$N =d_T-1$ erasures in any sliding window of length ${T+1}$.
Conversely there exists at-least one erasure pattern with $d_T$ erasures in a window of length ${T+1}$ where the decoder fails to recover all source packets.
\label{fact:dT}
\end{fact}

To the best of our knowledge the column span of a convolutional code was first introduced in~\cite{MartinianS04}
in the context of low-delay codes for burst erasure channels. 

\begin{defn}[Column Span]
The column span of $\bG^s_T$ is defined as

$$c_T  = \min_{\substack{\bs \equiv [\bs_0, \bs_1, \ldots, \bs_T]\\ \bs_0 \neq 0}} \mathrm{span}(\bs \cdot \bG^s_T)$$
where $\mrm{span}(\bx)$ computes the length of the support of the vector $\bx$ i.e., 
$\mrm{span}(\bx) = j-i+1,$
where $j$ is the last index where $\bx$ is non-zero and $i$ is the first such index.
\end{defn}

\begin{fact}
A necessary and sufficient condition for a  convolutional code  to recover every erased symbol with a delay of $T$ from a channel
that introduces no more than a single erasure burst of maximum length $B$ in any sliding window of length ${T+1}$ is that $c_T > B$.
\label{fact:cT}
\end{fact}

We omit a justification of these results due to space constraints.

It follows from Facts~\ref{fact:dT} and~\ref{fact:cT} that a necessary and sufficient condition for any convolutional code
to recover each source packet with a delay of $T$  over a channel that introduces either $N$ arbitrary erasures or $B$ consecutive
erasures in a sliding window of length ${T+1}$ is that $d_T > N$ and $c_T > B$.  Thus it is of interest to investigate code constructions that simultaneously have
a large column distance and a large column span. 

It turns out that large column-distance and large column-span are conflicting requirements in general. The following Theorem provides an outer-bound on the set of all achievable pairs $(c_T, d_T)$ for any code of a given rate. 

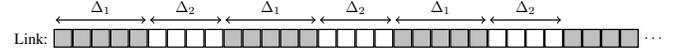
\begin{figure}
	\centering
	\resizebox{\columnwidth}{!}{
	\begin{tikzpicture}[node distance=0mm]
		\node                       (x1start) {Link:};
		\node[esym, right = of x1start]  (x100) {};
		\node[esym, right = of x100]     (x101) {};
		\node[esym, right = of x101]     (x102) {};
		\node[esym, right = of x102]     (x103) {};
		\node[esym, right = of x103]     (x104) {};
		\node[usym, right = of x104]     (x105) {};
		\node[usym, right = of x105]     (x106) {};
		\node[usym, right = of x106]     (x107) {};
		\node[usym, right = of x107]     (x108) {};
		\node[esym, right = of x108]     (x109) {};
		\node[esym, right = of x109]     (x110) {};
		\node[esym, right = of x110]     (x111) {};
		\node[esym, right = of x111]     (x112) {};
		\node[esym, right = of x112]     (x113) {};
		\node[usym, right = of x113]     (x114) {};
		\node[usym, right = of x114]     (x115) {};
		\node[usym, right = of x115]     (x116) {};
		\node[usym, right = of x116]     (x117) {};
		\node[esym, right = of x117]     (x118) {};
		\node[esym, right = of x118]     (x119) {};
		\node[esym, right = of x119]     (x120) {};
		\node[esym, right = of x120]     (x121) {};
		\node[esym, right = of x121]     (x122) {};
		\node[usym, right = of x122]     (x123) {};
		\node[usym, right = of x123]     (x124) {};
		\node[usym, right = of x124]     (x125) {};
		\node[usym, right = of x125]     (x126) {};
		\node[esym, right = of x126]     (x127) {};
		\node[esym, right = of x127]     (x128) {};
		\node[esym, right = of x128]     (x129) {};
		\node[esym, right = of x129]     (x130) {};
		\node      [right = of x130]     (x1end) {$\cdots$};

		\dimup{x100}{x104}{2mm}{$\Delta_1$};
		\dimup{x105}{x108}{2mm}{$\Delta_2$};
		\dimup{x109}{x113}{2mm}{$\Delta_1$};
		\dimup{x114}{x117}{2mm}{$\Delta_2$};
		\dimup{x118}{x122}{2mm}{$\Delta_1$};
		\dimup{x123}{x126}{2mm}{$\Delta_2$};
	\end{tikzpicture}}
	\caption{The periodic erasure channel used to prove an upper bound on capacity in Theorem~\ref{thm:bnd}. Here ${\Delta_1 = c_T-1}$
	and ${\Delta_2 = T-d_T+2}$ holds. The shaded symbols are erased while the remaining ones are received by the destination.}
	\label{fig:PEC}
\end{figure}

\begin{figure}[htbp]
	\centering
	\resizebox{5cm}{!}{
	\begin{tikzpicture}[node distance=0mm]
		\node                       (x1start) {Link:};
		\node[esym, right = of x1start]  (x100) {};
		\node[esym, right = of x100]     (x101) {};
		\node[esym, right = of x101]     (x102) {};
		\node[esym, right = of x102]     (x103) {};
		\node[esym, right = of x103]     (x104) {};
		\node[esym, right = of x104]     (x105) {};
		\node[esym, right = of x105]     (x106) {};
		\node[esym, right = of x106]     (x107) {};
		\node[usym, right = of x107]     (x108) {};
		\node[usym, right = of x108]     (x109) {};
		\node[usym, right = of x109]     (x110) {};
		\node[usym, right = of x110]     (x111) {};
		\node[usym, right = of x111]     (x112) {};
		\node[usym, right = of x112]     (x113) {};
		\node      [right = of x113]     (x1end) {$\cdots$};

		\dimdn{x100}{x103}{-2mm}{$c_T-d_T$};
		\dimup{x104}{x107}{2mm}{$d_T-1$};
		\dimup{x108}{x113}{2mm}{$T -d_T+2$};
		\dimdn{x100}{x107}{-10mm}{$c_T-1$};
	\end{tikzpicture}}
	\caption{One period of the periodic erasure channel in Fig.~\ref{fig:PEC}.}
	\label{fig:period}
\end{figure}
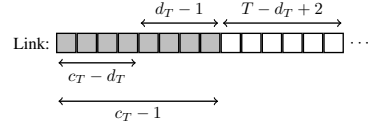

\begin{thm}[Column-Distance and Column-Span Tradeoff]
For any rate $R$ convolutional code with a column distance of $d_T$ and a column span of $c_T$, must satisfy :
\begin{align}
\left(\frac{R}{1-R}\right)c_T + d_T \le T + 1 + \frac{1}{1-R},
\label{eq:u-bnd}
\end{align}
as well as  $d_T \le c_T$ and $c_T \le T+1$.
\label{thm:bnd}
\end{thm}

\begin{proof}
We consider a periodic erasure channel with a period of ${P = T+c_T-d_T+1}$ and suppose that in every such period the first
${B=c_T-1}$ symbols are erased. We claim that for any convolutional code with a column-span and column-distance of $c_T$
and $d_T$ respectively, the decoder can reconstruct every source packet from such an erased sequence. 

Consider the first period that spans the interval $[0, P-1]$. The first $c_T-d_T$ erased symbols all need to be recovered
by time $t=P-1$. Thus in the window of interest, these symbols only experience a single erasure burst of length $c_T-1$  or smaller.
From Fact~\ref{fact:cT} these symbols can be recovered by any code with column span of $c_T$.

The next $d_T-1$ symbols have a deadline after time $P-1$. To recover $\rvs[t]$ for $t \in [c_T-d_T+1, c_T-1]$ observe that the length ${T+1}$
window  $\cW_t = [t, t+T]$ has two erasure bursts --- one at the start and one at the end of the interval.  As shown in Fig.~\ref{fig:period} each
such interval has a total of $T-d_T+2$ non-erased symbols. Thus the total number of erased symbols equals $T+1 - (T-d_T+2) = d_T-1$. 
From Fact~\ref{fact:dT}, a code with a column distance of $d_T$ can recover all of these symbols. 

Finally for $t \in [d_T, P-1],$ the recovery window $\cW_t = [t, t+T]$ only sees a single-erasure burst of length $c_T-1$ and hence the column
span of $c_T$ suffices to recover these symbols. 

Having recovered all the symbols in $[0,P-1]$ by their deadline, we can cancel their effect in all future parity checks and repeat the same argument
for every other period. Thus we can recover all erased symbols. Thus the rate of the code is upper bounded by the capacity of the periodic erasure channel
which results in
\begin{align}
R \le  1- \frac{c_T-1}{T+c_T-d_T+1}.
\label{eq:Rbnd}
\end{align}
Rearranging, this equation reduces to~\eqref{eq:u-bnd}. The upper bound $d_T \le c_T$ follows by observing that a code that corrects $d_T-1$ arbitrary erasures in a sliding window of length ${T+1}$ trivially corrects an erasure burst of the same length. The bound ${c_T \le T+1}$ simply follows from the definition.
\end{proof}

\begin{remark} 
Substituting ${R = \frac{1}{2}},$ the expression in~\eqref{eq:u-bnd} reduces to the following upper bound
\begin{equation}
c_T + d_T \le T+3. \label{eq:sumbnd}
\end{equation}
We conclude that the $R=1/2$ codes found via a computer search in~\cite[Section V-B]{MartinianS04} are indeed optimal
as they all satisfy~\eqref{eq:sumbnd}.
We next propose a family of codes that meet the upper bound~\eqref{eq:sumbnd} when $R=1/2$.
\end{remark}

\section{Embedded Random Linear Codes}
We introduce a construction that provides a flexible tradeoff between the column-distance and column-span
discussed in Section~\ref{sec:dist-col}.  This family includes codes with maximum column distance and maximum column 
span as special cases. Hence we discuss these special cases first. 

\subsection{Maximum Column-Distance Codes}
\label{sec:RLC}
As stated in Theorem~\ref{thm:bnd} we always have that $c_T \ge d_T$.  For the maximum column distance,  we $c_T = d_T$ in~\eqref{eq:u-bnd},
\begin{equation}
d_T \le 1 + (1-R)(T+1). \label{eq:dTbnd}
\end{equation}
The upper bound is the singleton-bound equivalent for convolutional codes~\cite[Chapter 3]{zigangirov}. 
The upper bound is achieved whenever the generator matrix $\bG_s^T$ in~\eqref{eq:GsT} has a full rank property i.e., any set of $k(T+1)$
columns are linearly independent.
By selecting the entries in $\bG_s^T$ from a sufficiently large finite field, we can satisfy this property with high probability.
We will refer to this construction as a {\em Random Linear Code} (RLC). 

\subsection{Maximum Column-Span Codes}
\label{sec:SCo}

\begin{figure}[t]
\resizebox{\columnwidth}{!}{
\begin{tikzpicture}[node distance=1mm,
  ]
  \tikzstyle{triple2} = [rectangle split, anchor=text,rectangle split parts=3]
  \tikzstyle{double2} = [rectangle split, anchor=text,rectangle split parts=2]
  \tikzstyle{triple} = [draw, rectangle split,rectangle split parts=3]
	\tikzstyle{double} = [draw, rectangle split,rectangle split parts=2]
	\tikzset{block/.style={rectangle,draw}}
	
	
	\node[triple2, minimum width=1.6cm] (start) {$B$ Symbols
    \nodepart{second}
      $T-B$ Symbols
    \nodepart{third}
     \tikz{\node[double2] {$B$ \nodepart{second}Symbols};}
  };
	
  \node[triple,  right = of start,fill=light-gray,minimum width=1.6cm] (p1) {$u[0]$
    \nodepart{second}
      $v[0]$
    \nodepart{third}
      \tikz{\node[double2] {$u[-T]$ \nodepart{second}$+p_v(v^{-1})$};}
  };
  
  \node[triple, right = of p1,fill=light-gray,minimum width=1.6cm] (p11) {$u[1]$
    \nodepart{second}
      $v[1]$
    \nodepart{third}
      \tikz{\node[double2] {$u[-T+1]$ \nodepart{second}$+p_v(v^{0})$};}
  };
  
  \node [block, right of=p11,minimum width=1.65cm, minimum height=2.68cm,node distance=1.94cm, fill=light-gray] (p2) {$\cdots$};
   
  \node[triple,  right = of p2,fill=light-gray,minimum width=1.6cm] (p3) {$u[B-1]$
    \nodepart{second}
      $v[B-1]$
    \nodepart{third}
      \tikz{\node[double2] {$u[-T+B-1]$ \nodepart{second}$+p_v(v^{B-2})$};}
  };
  
  \node[triple,  right = of p3,minimum width=1.6cm] (p4) {$u[B]$
    \nodepart{second}
      $v[B]$
    \nodepart{third}
      \tikz{\node[double2] {$u[-T+B]$ \nodepart{second}$+p_v(v^{B-1})$};}
  };
  
  \node [block, right of=p4,minimum width=1.65cm, minimum height=2.68cm,node distance=1.96cm] (p5) {$\cdots$};
  
  \node[triple,  right = of p5,minimum width=1.6cm] (p6) {$u[T-1]$
    \nodepart{second}
      $v[T-1]$
    \nodepart{third}
      \tikz{\node[double2] {$u[-1]+$ \nodepart{second}$p_v(v^{T-2})$};}
  };
  
  \node[triple, right = of p6,minimum width=1.6cm] (p7) {$u[T]$
    \nodepart{second}
      $v[T]$
    \nodepart{third}
      \tikz{\node[double2] {$u[0]+$ \nodepart{second}$p_v(v^{T-1})$};}
  };
  
  \bracedn{p1}{p3}{-2mm}{\footnotesize{Erased Packets}};
  \bracedn{p4}{p6}{-2mm}{\footnotesize{Used to recover $v[0],\cdots ,v[B-1]$}};
  \bracedn{p7}{p7}{-2mm}{\footnotesize{Recover $u[0]$}};

\end{tikzpicture}}
\caption{A window of $T+1$ channel packets showing the code construction of Streaming Codes (SCo). $v^{t}$ denotes the set of symbols $(v[t-T],\dots,v[t])$. }
\label{fig:SCo_Construction}
\end{figure}

\begin{figure}[t]
\resizebox{\columnwidth}{!}{
\begin{tikzpicture}[node distance=1mm,
  ]
  \tikzstyle{triple2} = [rectangle split,rectangle split parts=3]
  \tikzstyle{double2} = [rectangle split,rectangle split parts=2]
  \tikzstyle{triple} = [draw, rectangle split,rectangle split parts=3]
	\tikzstyle{double} = [draw, rectangle split,rectangle split parts=2]
	\tikzset{block/.style={rectangle,draw}}
	
	
	\node[triple2,minimum width=1.6cm] (start) {$u$ Symbols
    \nodepart{second}
      $v$ Symbols
    \nodepart{third}
     \tikz{\node[double2] {$u$ \nodepart{second}Symbols};}
  };
	
  \node[triple, right = of start,minimum width=1.6cm] (p1) {$u[0]$
    \nodepart{second}
      $v[0]$
    \nodepart{third}
      \tikz{\node[double2] {$p_u(u^{-\Delta})$ \nodepart{second}$+p_v(v^{-1})$};}
  };
  
  \node[triple, right = of p1,minimum width=1.6cm] (p11) {$u[1]$
    \nodepart{second}
      $v[1]$
    \nodepart{third}
      \tikz{\node[double2] {$p_u(u^{-\Delta+1})$ \nodepart{second}$+p_v(v^{0})$};}
  };
  
  \node [block, right of=p11,minimum width=1.6cm, minimum height=2.71cm,node distance=1.94cm] (p2) {$\cdots$};
   
  \node[triple, right = of p2,minimum width=1.6cm] (p3) {$u[\Delta-1]$
    \nodepart{second}
      $v[\Delta-1]$
    \nodepart{third}
      \tikz{\node[double2] {$p_u(u^{-1})$ \nodepart{second}$+p_v(v^{\Delta-2})$};}
  };
  
  \node[triple, right = of p3,minimum width=1.6cm] (p4) {$u[\Delta]$
    \nodepart{second}
      $v[\Delta]$
    \nodepart{third}
      \tikz{\node[double2] {$p_u(u^{0})$ \nodepart{second}$+p_v(v^{\Delta-1})$};}
  };
  
  \node [block, right of=p4,minimum width=1.6cm, minimum height=2.71cm,node distance=1.95cm] (p5) {$\cdots$};
  
  \node[triple, right = of p5,minimum width=1.6cm] (p6) {$u[T-1]$
    \nodepart{second}
      $v[T-1]$
    \nodepart{third}
      \tikz{\node[double2] {$p_u(u^{T-\Delta-1})$ \nodepart{second}$+p_v(v^{T-2})$};}
  };
  
  \node[triple, right = of p6,minimum width=1.6cm] (p7) {$u[T]$
    \nodepart{second}
      $v[T]$
    \nodepart{third}
      \tikz{\node[double2] {$p_u(u^{T-\Delta})$ \nodepart{second}$+p_v(v^{T-1})$};}
  };
  

\end{tikzpicture}}
\caption{A window of $T+1$ channel packets showing the code construction of Embedded Random Linear Codes (E-RLC).}
\label{fig:ERLC_Construction}
\end{figure}

Clearly any convolutional code with a column span of $c_T \ge 2$ is guaranteed to have $d_T \ge 2$.
The later simply implies that at-least one erasure can be corrected in a window of length ${T+1}$.
Substituting $d_T=2$ in~\eqref{eq:u-bnd} and using $c_T \le {T+1}$, 
\begin{align}
c_T  \le 1 + T\cdot\min\left(\frac{1}{R}-1, 1\right).\label{eq:cTbnd}
\end{align}

A class of codes, SCo with this property is constructed in~\cite{MartinianT07, MartinianS04}. Due to space constraints do not review the code construction but refer the reader to~\cite{MartinianT07, MartinianS04} Instead, we describe a related construction that also achieves the maximum column span. The advantage of this construction is that it generalizes to constructions that simultaneously have large column span and column distance. This construction is illustrated in Fig.~\ref{fig:SCo_Construction} and the main steps are as described below.

\subsubsection*{Encoding}

\begin{enumerate}
\item Split each source symbol into a total of $T$ sub-symbols over ${\mathbb F}_q$, belonging to two groups as shown below.
\begin{align}
\bs[i] = \bigg\{\underbrace{\rvu_0[i],\ldots, \rvu_{B-1}[i]}_{=\bu[i]}, \underbrace{\rvv_0[i],\ldots, \rvv_{T-B-1}[i]}_{=\bv[i]}\bigg\}
\end{align}
\item Apply a $(T, T-B)$ systematic random linear code to the source symbols $\bv[i]$ and generate $B$ parity checks $\bp_v[i] = (p_0[i],\ldots, p_{B-1}[i])$
at time $i$ i.e.,
\begin{align}
\bp_v[i] = \sum_{j=1}^{T-1} \bv[i-j]\cdot \bG_j \label{eq:v-rlc}
\end{align}
where $\bG_j \in {\mathbb F}_q^{{T-B} \times B}$. It can be  verified from~\eqref{eq:dTbnd} that such a code can recover up-to $B$ erasures in a window of length $T$.

\item Apply a repetition code to $\bu[i]$ with a delay of $T$ and then combine them with $\bv[i]$ i.e.,

\begin{align}
\bx[i] = \left(\begin{array}{c}\bu[i] \\ \bv[i] \\ \bp_v[i] \oplus \bu[i-T]\end{array}\right). \label{eq:sco}
\end{align}
\end{enumerate}
Suppose that an erasure burst spans $t \in [0, B-1]$ (c.f.~Fig.~\ref{fig:SCo_Construction}). The receiver needs to recover $\bs[j]$ by time $j+T$ for $j \in \{0,\ldots, B-1\}$.
Our proposed decoder uses the parity checks of the random linear code to first recover all the symbols in $\bv[j]$ simultaneously by time ${T-1}$.
Having recovered these symbols the decoder sequentially recovers the symbols $\bu[j]$ at time $j+T$ using the repetition code.  More specifically
the decoder implements the following steps.

\subsubsection*{Decoding}
\begin{itemize}
\item Recover the parity checks symbols $\bp_v[B],\ldots, \bp_v[T-1]$ from $\bx[B],\ldots, \bx[T-1]$ by cancelling the symbols $\bu[t]$ for $t< 0$ that are not erased.
\item Recover the symbols $\bv[0],\ldots, \bv[B-1]$ from parity checks $\bp_v[B],\ldots, \bp_v[T-1]$ using random linear code~\eqref{eq:v-rlc}.
\item For $j \in [0,B-1]$, at time $j+T$, first compute the parity check $\bp_v[j+T]$ which is a function of symbols $\bv[i]$ that have been recovered already and then
subtract it from $\bu[j] + \bp_v[j+T]$ to recover $\bu[j]$. Thus the source symbol $\bs[j] = (\bu[j], \bv[j])$ is recovered by time ${j+T}$  although the symbols $\bv[j]$
is recovered by time ${T-1}$.
\item All the erased symbols are recovered by time $t=T+B-1$. The encoder can recover from a second erasure-burst starting at time $t=T+B$ or later. 
This is equivalent to the condition that $c_{T}= B+1$.
\end{itemize}

Notice that the proposed construction takes a RLC code over $\bv[\cdot]$ as a base code and {\em embeds} additional symbols $\bu[\cdot]$. The parity
checks of $\bu[\cdot]$ are simple repetition codes and directly combined with  $\bp_v[\cdot]$ after a shift of $T$. Thus the rate  increases over the base RLC code 
upon addition of $\bu[\cdot]$. In the generalization of this construction we replace the repetition code with another RLC code. 

\subsection{Proposed Construction}
The use of a repetition code in the previous section limits the column distance to $d_T=2$.
To improve the column distance we first replace the repetition code for $\bu[\cdot]$ with another random linear code. Furthermore
instead of applying a shift of $T$ to the parity checks of the $\bu[\cdot]$ symbols we apply a shift of $\Delta \le T$. In particular we construct the parity checks $\bp_v[i]$ as in~\eqref{eq:v-rlc}
and construct a second set of parity checks 
\begin{align}
\bp_u[i] = \sum_{j=0}^{T-\Delta} \bu[i-j] \bH_j.
\end{align}


We will assume that $\bu \in \mathbb{F}_q^u$ and $\bv \in \mathbb{F}_q^v$ and the parity checks $\bp \in \mathbb{F}_q^u$. 

\begin{align}
\bx[i] = \left(\begin{array}{c}\bu[i] \\ \bv[i] \\ \bp_v[i] \oplus \bp_u[i-\Delta]\end{array}\right). \label{eq:layered}
\end{align}
We will assume that the entries in the matrices of $\bH_j$ and $\bG_i$ are all sampled uniformly at random and $q$ is
sufficiently large so that all the sub-matrices of interest have either full row-rank or column-rank with high probability.
The code construction is illustrated in Fig.~\ref{fig:ERLC_Construction}.
The rate of the code is given by, \begin{align}R = \frac{u+v}{2u+v}.\label{eq:r-uv}\end{align}

We develop closed form expressions for the column span and column distance of the proposed code construction below.
\begin{prop}
The column span of the Embedded-Random Linear Code with a shift of $\Delta$ is given by
\begin{equation}
c_T = \begin{cases}
\frac{1-R}{R}\Delta +1, & R \le \frac{\Delta}{T+1} \\
(1-R)(T+1) +1, & R > \frac{\Delta}{T+1}
\end{cases}
\label{eq:cT-thresh}
\end{equation}
\label{prop:cT-thresh}
\end{prop}
\begin{proof}
To compute the column span it is sufficient to find  largest erasure burst length $B$ starting at time $t=0,$ such that $\bs[0]$ can be recovered by time $t=T$.
Note that the parity checks $p_\bu[\cdot]$ involving $\bu[0],\ldots \bu[B-1]$ appear from time $t=\Delta, \ldots, \Delta+B-1$. The parity checks in the interval $[B, \Delta-1]$
do not involve any $\bu[\cdot]$ symbols that are erased.

We first find the condition under which the parity-checks in the interval $[B,\Delta-1]$ can be used to recover all the $\bv[\cdot]$ symbols in time $[0, B-1]$.  Since there are a 
total of $\Delta-B$ parity check symbols, each contributing $u$ equations and a total of $B$ erased symbols, each generating $v$ unknowns we must have that
\begin{equation}
B\cdot v \le (\Delta-B) u\label{eq:B-seq-recovery}
\end{equation}
which implies from~\eqref{eq:r-uv} that $B\le \frac{1-R}{R}\Delta$. Once all the erased $\bv[\cdot]$ symbols are recovered, their contribution can be cancelled
from future parity checks and the symbol $\bu[0]$ can be recovered at time $\Delta\le T$. 

If~\eqref{eq:B-seq-recovery} is not satisfied  then all the $\bu[\cdot]$ and $\bv[\cdot]$ symbols need to be simultaneously recovered at time $t=T$.
There are a total of ${T+1-B}$ non-erased parity check symbols in the interval $[B,T]$ and each parity check contributes $u$ equations. The total number of unknowns
from the $B$ erased symbols is $B(u+v)$. Thus we must have that
\begin{equation}
B(u+v) \le (T+1-B)u
\end{equation}
which leads to $B \le(1-R)(T+1)$. Thus it follows that the maximum burst length that can be corrected is given by
\begin{align}
B = \max\bigg\{(1-R)(T+1), \frac{1-R}{R}\Delta\bigg\},
\end{align}
from which the claim easily follows.
\end{proof}

\begin{remark}
Our result in Prop.~\ref{prop:cT-thresh} shows that to improve the column span over a random linear code one must take the 
shift to satisfy $\Delta \ge R\cdot(T+1)$. 
\end{remark}

\begin{prop}
\label{prop:dT-thresh}
The column-distance of the Embedded-Random Linear Code with a shift of $\Delta \ge R(T+1) $ and $R \ge \frac{1}{2}$ is given by the following
\begin{align}
d_T= \frac{1-R}{R}(T-\Delta) +2\label{eq:dT-thresh}
\end{align}
\end{prop}
\begin{proof}

We need to show that for any erasure sequence in the window $[0,T]$ if the symbol $\rvs[0]$ is not recovered by time $t=T$ then the number of erasures must be at-least $d_T$. 

We first observe that if the symbol $\bv[0]$ is not recovered by time ${t=\Delta-1}$ then the code behaves like a random linear code. Every parity check sub-symbol provides one independent equation. The total number of erasures necessary is given by the column distance of the random linear code~\eqref{eq:dTbnd}, which is the maximum possible column distance and exceeds~\eqref{eq:dT-thresh}.

Thus we only need to consider those erasure patterns where $\bv[0]$ is recovered by time ${t=\Delta-1}$. In addition to $s[0]$ we consider three groups of symbols. Group $1$ consists of  $k_1$ symbols that are erased in time $t \in [1,\Delta-1]$ such that the corresponding $\bv[\cdot]$ is recovered by time ${t=\Delta-1}$. Group $2$ consists of $k_2$ symbols erased in the same interval whose $\bv[\cdot]$ symbols are not recovered by time ${t=\Delta-1}$. Group $3$ consists of $k_3$  symbols erased in the time $t \in [\Delta,T]$. We seek the minimum possible value of ${k_1 + k_2 + k_3}$ such that the symbol $\bu[0]$ is not recovered by time $t=T$.

Since $R \ge \frac{1}{2}$ and $\Delta \ge R(T+1)$ we have that $2\Delta \ge T+1$. Thus the $\bu[\cdot]$ symbols of group $3$ are involved in parity checks after time $t=T$ and hence do not need to be considered. 

We consider a possibly sub-optimal decoder that attempts to recover the remaining symbols using only the parity checks in the interval $t \in [\Delta,T]$. Clearly, by using such a sub-optimal decoder we can only under-estimate the number of erasures that can be corrected. Since the symbol $\bu[0]$ start appearing in the parity checks starting at time $t=\Delta$ (c.f.~\eqref{eq:layered}) and is not recovered by time $t=T$ (by assumption), each of the parity checks sub-symbols
in the interval $[\Delta,T]$ provides one non-redundant equation. Thus it follows that the total number of unknown associated with the remaining erased symbols must exceed the number of available parity check equations. The total number of unknowns is upper bounded by a sum of three terms:
\begin{itemize}
\item The $\bu[\cdot]$ symbols in group $1$ and $\rvs[0]$: $N_1=(k_1+1)u$
\item Both $\bu[\cdot]$ and $\bv[\cdot]$ symbols in group 2: $N_2 = k_2(u+v)$
\item The $\bv[\cdot]$ symbols in group $3$: $N_3 = k_3\cdot v$
\end{itemize}
The total number of available equations from the parity checks in the interval $[\Delta,T]$ where there are $k_3$ erasures is given by
$(T-\Delta-k_3+1)u$. Thus a necessary condition under which $\bu[0]$ is not recovered  is given by:
\begin{align}
N_1 + N_2 + N_3 \ge (T-\Delta-k_3+1)u
\end{align}
Upon substituting for $N_i$ and through some simple algebra we get that
\begin{align}
(k_1 + k_2 + k_3)u + (k_2 + k_3)v \ge (T-\Delta)u
\end{align}
which in turn implies that
\begin{align}
k_1 + k_2 + k_3 \ge \frac{u}{u+v}(T-\Delta) = \frac{1-R}{R}(T-\Delta)
\end{align}
Thus the total number of erasures in any such sequence must exceed $1+ \frac{1-R}{R}(T-\Delta)$ as stated in~\eqref{eq:dT-thresh}.
\end{proof}

\begin{remark}
For the special case of $R=1/2$ and ${\Delta \ge \frac{T+1}{2}}$ note that from Prop.~\ref{prop:cT-thresh} that $c_T = \Delta +1$. From Prop.~\ref{prop:dT-thresh} we have that $d_T = T-\Delta +2$. Thus we have that $d_T+c_T = T+3,$ which meets the upper bound in
~\eqref{eq:sumbnd}. Thus the proposed embedded-random linear code constructions provide a family of codes that are optimal for $R=1/2$.
As discussed before, the embedded-random linear codes are also optimal in the special case of maximum column span or maximum column distance. Their optimality in other cases remains to be seen. The gaps between the upper-bound given in~\eqref{eq:sumbnd} and values achieved by Embedded-RLC codes are illustrated in Fig.~\ref{fig:cTdT_Tradeoff}.
\end{remark}

\begin{figure}
\centering
\includegraphics[width=\linewidth]{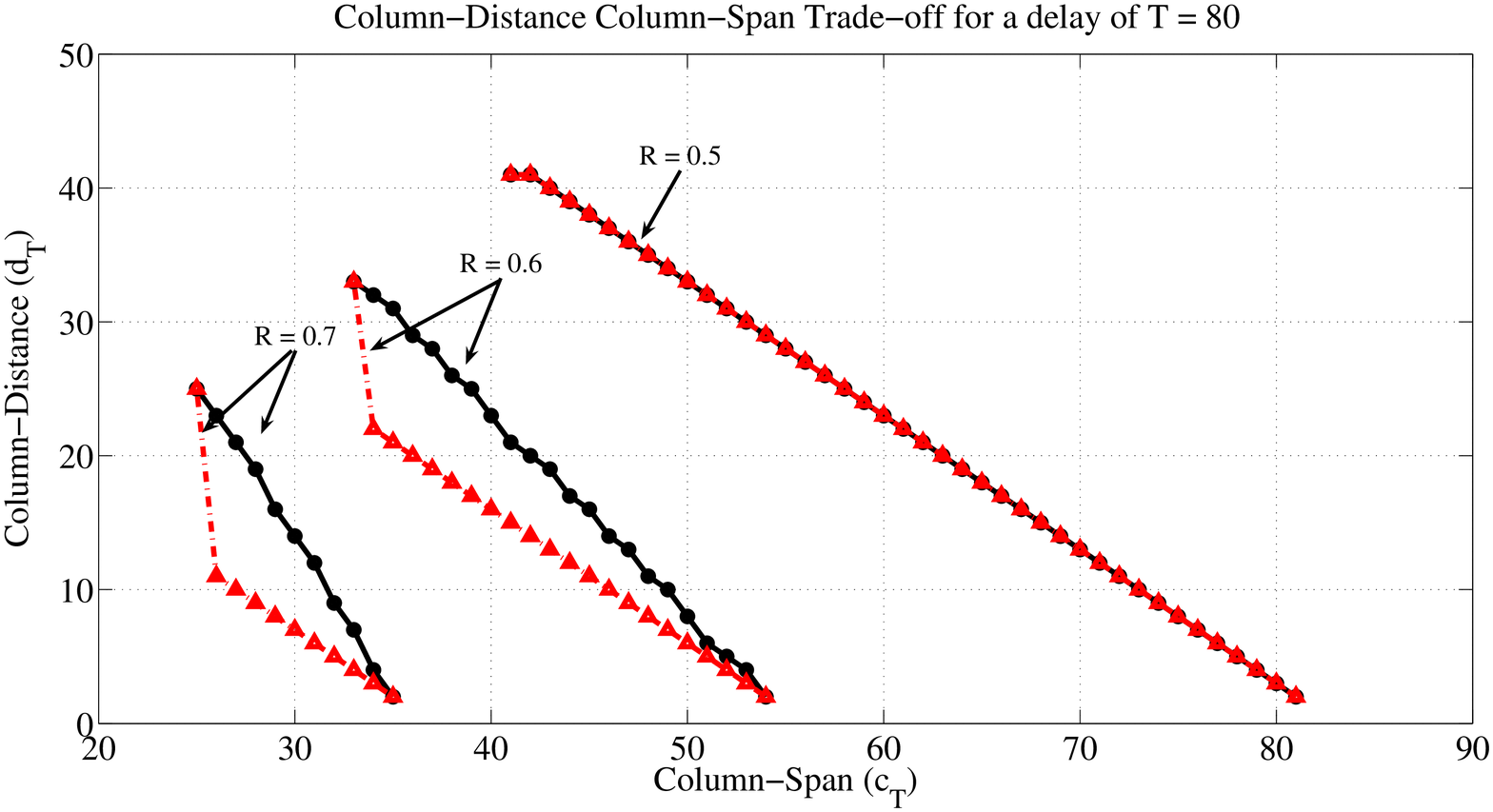}
\caption{Comparison of Upper-Bounds (solid black lines) and Lower-Bounds achieved by Embedded-RLC codes (broken red lines) for Column-Distance Column-Span Trade-off. Each pair correspond to different rates of $R = 0.5$, $0.6$ and $0.7$ from right to left but a fixed decoding delay of $T=80$ symbols.}
\label{fig:cTdT_Tradeoff}
\end{figure}

\section{Simulation Results - Gilbert-Elliott Channel Model}

\begin{figure*}
  \begin{minipage}[b]{0.5\linewidth}
    \centering
    \includegraphics[width=\linewidth]{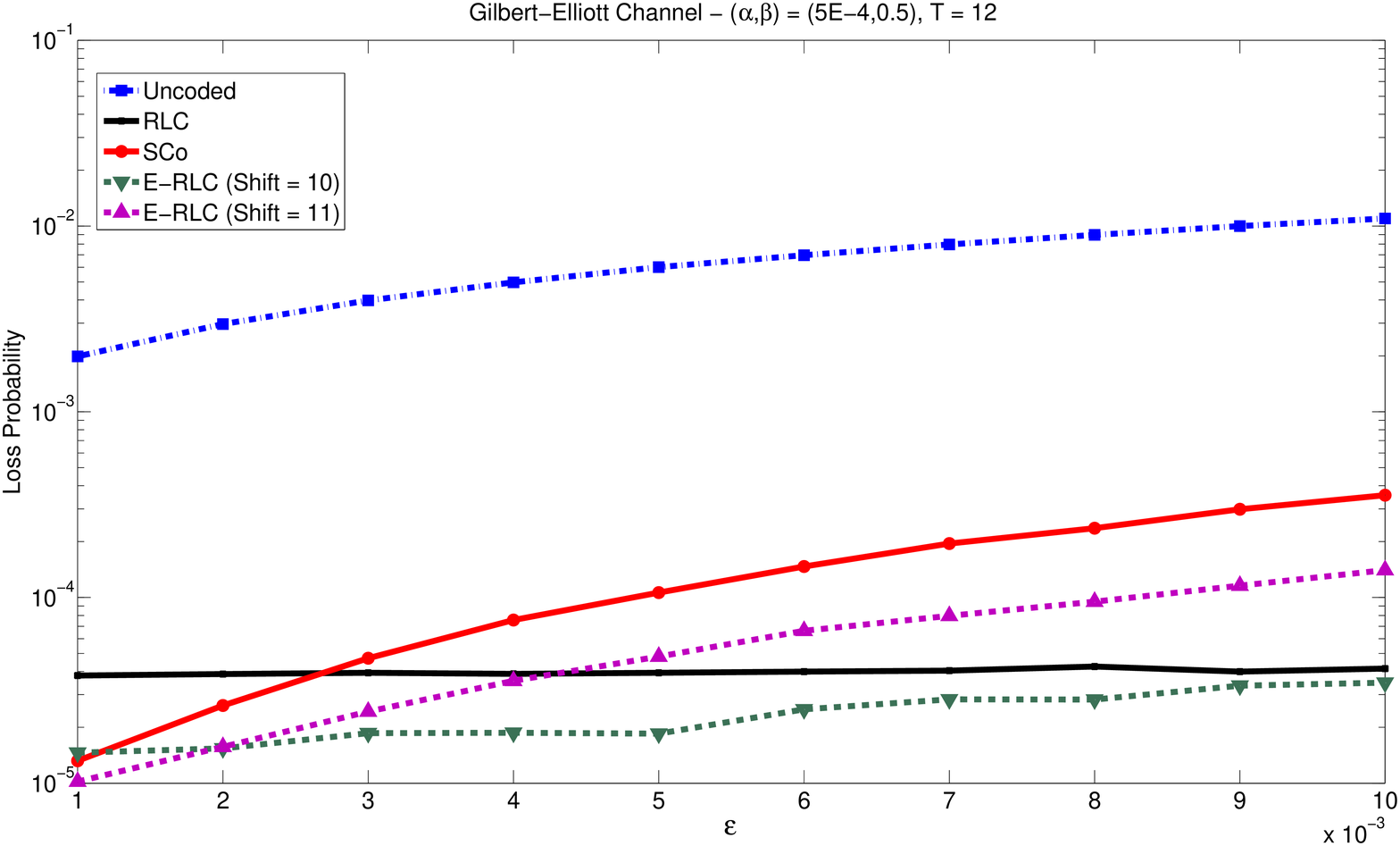}
    \caption{Simulation over a Gilbert-Elliott Channel with $(\al,\beta) = (5 \times 10^{-4},0.5)$. All codes are evaluated using a decoding delay of $T=12$ symbols.}
    \label{fig:Gilbert_T12}
  \end{minipage}
  \hspace{0.5cm}
    \begin{minipage}[b]{0.5\linewidth}
    \centering
\includegraphics[width=\linewidth]{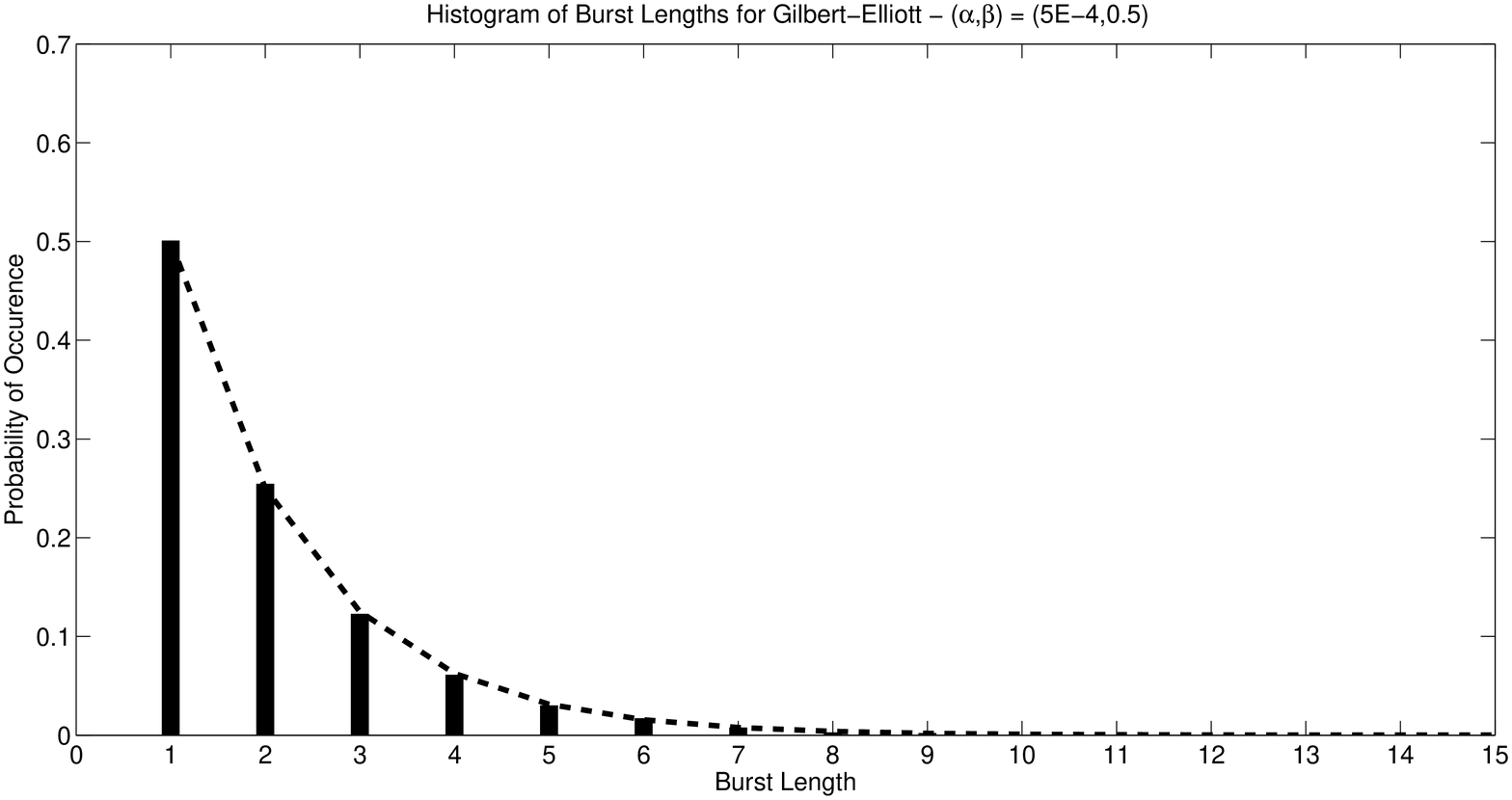}
\caption{Histogram of Bursts when $\beta =0.5$ which approximates a geometric distribution (shown dotted) with success probability of $0.5$.}
\label{fig:Gilbert_T12_Burst}
  \end{minipage}
\end{figure*}

\begin{figure*}
  \begin{minipage}[b]{0.5\linewidth}
    \centering
    \includegraphics[width=\linewidth]{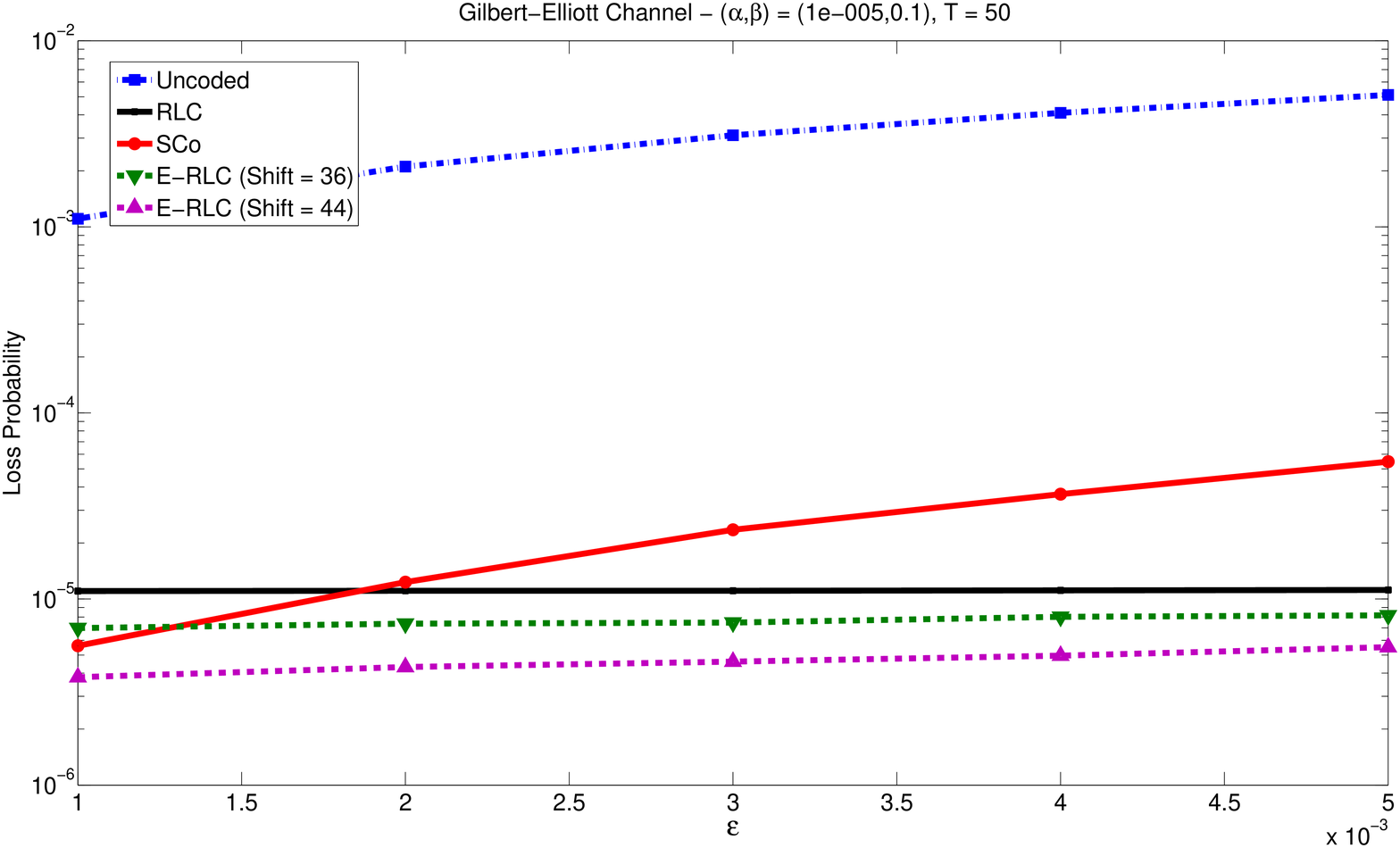}
    \caption{Simulation over a Gilbert-Elliott Channel with $(\al,\beta) = (10^{-5},0.1)$. All codes are evaluated using a decoding delay of $T=50$ symbols.}
    \label{fig:Gilbert_T50}
  \end{minipage}
  \hspace{0.5cm}
    \begin{minipage}[b]{0.5\linewidth}
    \centering
\includegraphics[width=\linewidth]{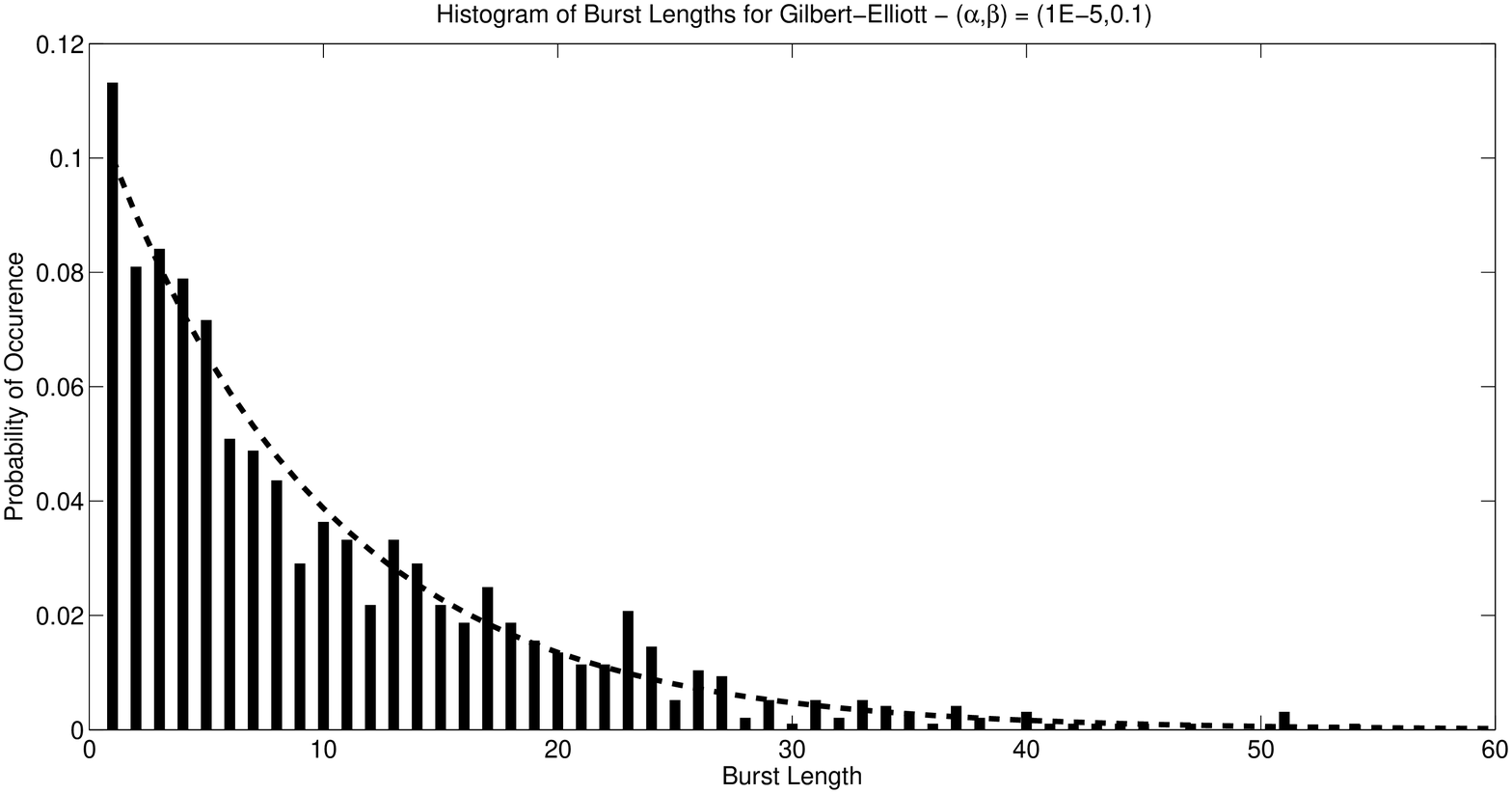}
\caption{Histogram of Bursts when $\beta =0.1$ which approximates a geometric distribution (shown dotted) with the same success probability.}
\label{fig:Gilbert_T50_Burst}    
  \end{minipage}
\end{figure*}

We consider a two-state Gilbert-Elliott channel model~\cite{elliott,gilbert}. 
In the ``good state" each channel packet is lost with a probability of $\eps$ whereas in the ``bad state" each channel packet is lost with a probability of $1$. We note that the average loss rate of the Gilbert-Elliott channel is given by
\begin{align}
\Pr(\cE) = \frac{\beta}{\beta+\al}\eps + \frac{\al}{\al+\beta} \label{eq:loss-uncoded}.
\end{align}
where $\al$ and $\beta$ denote the transition probability from the good state to the bad state and vice versa. 

 



As long as the channel stays in the bad state the channel behaves as a burst-erasure
channel. The length of each burst is a Geometric random variable with mean of $\frac{1}{\beta}$.
When the channel is in the good state it behaves as an i.i.d.\ erasure channel with an erasure
probability of $\eps$. The gap between two successive bursts is also a geometric random variable
with a mean of $\frac{1}{\al}$.

Fig.~\ref{fig:Gilbert_T12} and Fig.~\ref{fig:Gilbert_T50} show the simulation performance over a Gilbert-Elliott Channel.
The parameters chosen in the two plots are as shown in Table~\ref{tab:GE-Params}.

\begin{table}[!htb]
	\centering
		\begin{tabular}{l|c|c}\\\hline
			& Fig.~\ref{fig:Gilbert_T12} & Fig.~\ref{fig:Gilbert_T50} \\\hline
			Delay $T$ & 12 & 50 \\
			$(\alpha,\beta)$ & $(5 \times 10^{-4},0.5)$ & $(10^{-5},0.1)$ \\
			Channel Length & $10^7$ & $10^8$ \\
			Rate $R$ & $12/23$ & $50/99$ \\\hline
		\end{tabular}
	\caption{Gilbert-Elliott Channel Parameters}
	\label{tab:GE-Params}
\end{table}

We note that the channel parameters for the $T=12$ case are the same as those used in~\cite[Section 4-B, Fig.~5]{MartinianS04}. For the case when $\eps = 0$ we have verified   that our simulations agree with the results in~\cite{MartinianS04}. Note that we do not  use $R=0.5$, because the SCo codes degenerate into simple repetition codes for this case~\cite{MartinianT07}. We use the next highest rate for each class of codes. The choice of $\beta$ is smaller for $T=50$ because we expect to be able to correct longer bursts because of the larger delay. The histogram of burst lengths for both channels is shown in Fig.~\ref{fig:Gilbert_T12_Burst} and Fig.~\ref{fig:Gilbert_T50_Burst}. The choice of $\al$ is taken to be sufficiently small so that the contribution from failures due to small guard periods between bursts is not dominant. Note that our proposed constructions degenerate to RLC codes when the inter-burst gaps are smaller than the decoding delay and hence the performance gains are not observed in that regime.




\begin{table}[!htb]
	\centering
		\begin{tabular}{c|c|c|c}
		$(R,T)$& Shift $(\Delta)$& Col.-Span $(c_T)$ & Col.-Distance $(d_T)$ \\\hline \hline
		$(\frac{12}{23},12)$ & 10 & 10& 3 \\
		& 11 & 11& 2 \\\hline
		$(\frac{50}{99},50)$ & 36 & 36& 15 \\
		& 44 & 44 & 7 \\\hline
		$(\frac{40}{79},40)$ & 32 & 32 & 9 \\
		& 36 & 36 & 5 \\\hline
		$(\frac{80}{159},80)$ & 48 & 48 &33 \\
		& 52 & 52 &29 \\
		& 60 & 60 &21 \\\hline
		\end{tabular}
		\caption{Column Distance and Span for Embedded-RLC Codes for ome Rates $R$ and Delays $T$.}
		\label{tab:ERLCProperties}
		\vspace{-2em}
\end{table}


In Fig.~\ref{fig:Gilbert_T12} and Fig.~\ref{fig:Gilbert_T50} we observe that our Embedded-RLC constructions provide improved error correction capability over both the SCo codes and RLC codes by virtue of their longer column span and column distance.  We discuss the performance of various codes in more detail below.

\begin{itemize}

\item {\bf Uncoded Loss Rate}: The uppermost plot in Fig.~\ref{fig:Gilbert_T12} and Fig.~\ref{fig:Gilbert_T50} is the uncoded packet loss rate. It agrees well with the expression in~\eqref{eq:loss-uncoded}.
\item {\bf Random Linear Codes}: The solid horizontal black line is the loss-rate of the Random Linear Code (RLC) in Section~\ref{sec:RLC} which has the maximum column distance.
We see that for the range of $\eps$ that we consider the RLC is able to correct all the erasures in the good state and hence the loss rate does not depend on $\eps$. The only losses that occur are when the burst-lengths in the bad state exceed $B =6$ in Fig.~\ref{fig:Gilbert_T12} and $B=25$ in Fig.~\ref{fig:Gilbert_T50}. The loss rates  of $\Pr(\cE) \approx 4 \times 10^{-5}$ and $\Pr(\cE) \approx 10^{-5}$ observed in the two cases are consistent with the probability of observing such long bursts.

\item {\bf Streaming Codes}: The SCo Codes are represented by the red plot. We see that in the interval of $\eps$ considered, there is a noticeable increase in the loss rate. The performance is better than RLC codes for $\eps \approx 10^{-3}$ but deteriorates quickly as we increase $\eps$. The packet-loss probability increases in proportion to $\eps^2$ as $d_T=2$ for these codes. 


\item {\bf Embedded-Random Linear Codes}:  The associated column-distance and column-span of these codes from Proposition~\ref{prop:cT-thresh} and~\ref{prop:dT-thresh} are indicated in Table~\ref{tab:ERLCProperties}.  For $T=12$ case, the performance of the Embedded-Random Linear Codes with shifts of $\Delta \in \{10,11\}$ is shown in Fig.~\ref{fig:Gilbert_T12}. The shift of $\Delta =11$ also has a column distance of $2$. It follows a similar trend as SCo codes and its performance deteriorates quickly with $\eps$. The shift of $\Delta = 10$ provides the best performance in Fig.~\ref{fig:Gilbert_T12}. This code has a column-distance of $d_T = 3$. We observe that the effect of i.i.d. erasures is not significant for most of the interval of $\eps$ considered as the loss rate scales as $\eps^3$. These code do have a smaller column span than SCo codes and hence its performance is slightly worse in the other extreme of $\eps\approx 10^{-3}$.   For $T=50$ case, the Embedded-RLC codes with shifts of $\Delta \in \{36, 44\}$ are shown in Fig.~\ref{fig:Gilbert_T50}. Since these codes have a column distance of at-least $7$ the performance does not deteriorate as noticeably as the SCo codes in the range of $\eps$ of interest. The shift of $44$ has the best performance because it has a longer column-span and hence can correct longer erasure bursts.
\end{itemize}

\section{Simulation Results - Fritchman Channel Model}
In this section, we consider a special class of Fritchman Channel Model~\cite{fritchman} with a total of ${N+1}$ states. One of the states is the error free state and the remaining $N$ states are error states. Fritchman  and related higher order Markov models are commonly used to model fade-durations in mobile links.
\begin{figure}
\includegraphics[width=\linewidth]{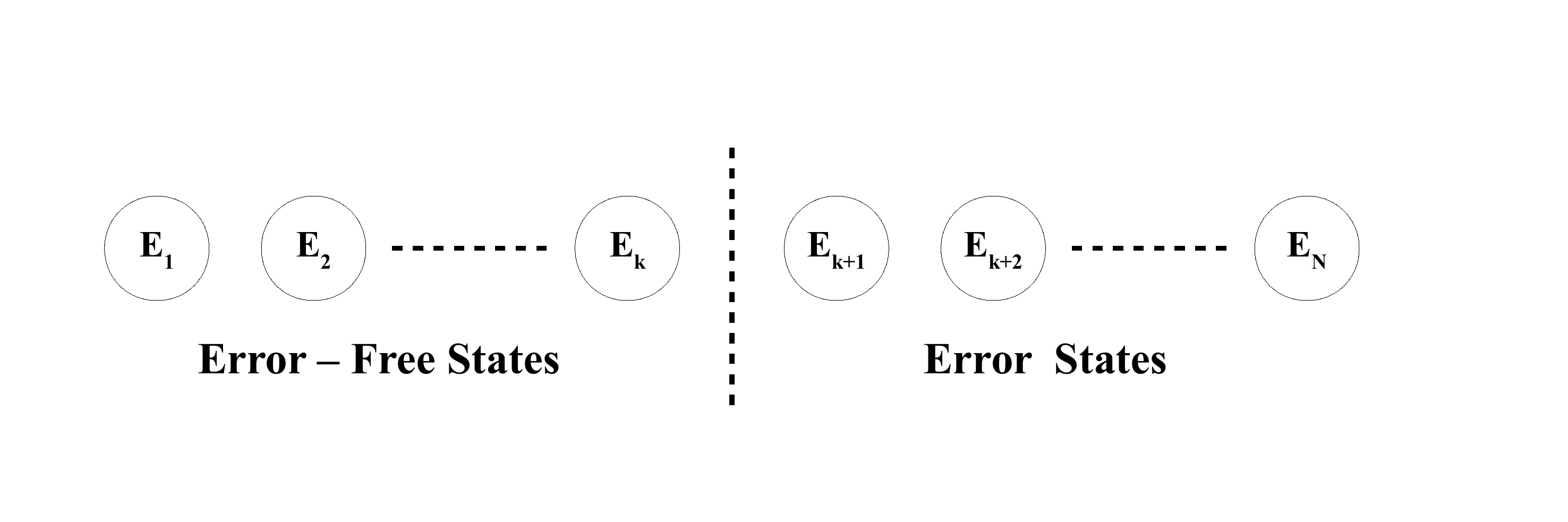}
\caption{The Fritchman Model with One Good State and $N$ Error States. In each error state the packet is lost with probability $1$ whereas in the good state it is lost with probability $\eps$.}
\end{figure}

We let the transition probability from the good state to the first error state $E_1$ to be $\al$ whereas the transition probability from each of the error states equals $\beta$. Let $\eps$ be the probability of a packet loss  in good state. We lose packets in any error state with probability $1$.  We consider two scenarios in Fig.~\ref{fig:Fritchman_T40} and~\ref{fig:Fritchman_T80} whose parameters are shown in Table~\ref{tab:Fritchman}.

\begin{table}[!htb]
	\centering
		\begin{tabular}{l|c|c}\\\hline
			& Fig.~\ref{fig:Fritchman_T40} & Fig.~\ref{fig:Fritchman_T80} \\\hline
			Channel States & 9 & 20 \\
			Delay $T$ & 40 & 80 \\
			$(\alpha,\beta)$ & $(10^{-5},0.5)$ & $(10^{-5},0.5)$ \\
			Channel Length & $10^8$ & $10^8$ \\
			Rate $R$ & $40/79$ & $80/159$ \\\hline
		\end{tabular}
	\caption{Fritchman Channel Parameters}
	\label{tab:Fritchman}
\end{table}

Fig.~\ref{fig:Fritchman_T40_Burst} and Fig.~\ref{fig:Fritchman_T80_Burst} illustrate the empirical histogram of burst-lengths in a sample erasure pattern generated over a channel of $10^8$ symbols. The actual distribution is given by a negative binomial distribution and is shown by the dotted envelope. 

In both Fig.~\ref{fig:Fritchman_T40} and Fig.~\ref{fig:Fritchman_T80}, the uncoded loss rate is shown by the upper-most plot while the black horizontal line is the performance of RLC.  
 Note that the  performance of RLC is essentially independent of $\eps$ in the interval of interest. As before the RLC codes clean up all the losses in the good state and fail against burst lengths longer than its column span. The performance of the SCo codes is shown by the red-plot in both figures. We note that it is better than the RLC code for $\eps = 10^{-3}$ but deteriorates quickly as we increase $\eps$.  There are two dominant error events for SCo codes. One is the simultaneous erasure of the symbols in the repetition code. The second is the occurrence of an isolated erasure in the good state  in the interval of length $T$ following a transition from the bad state. This particular event is significant for larger values of $T$.

The parameters of the embedded-RLC codes used in these figures are shown in Table~\ref{tab:ERLCProperties}. In Fig.~\ref{fig:Fritchman_T40} we observe that the shift of $\Delta=32$ has the smallest loss-rate over the interval of $\eps$ of interest. The longest burst-length observed in Fig.~\ref{fig:Fritchman_T40_Burst} is $B=30,$ which can be recovered by this shift and the relatively larger column distance of $d_T=9$ makes it more resilient than
the shift of $\Delta=36$. In Fig.~\ref{fig:Fritchman_T80} we observe that the shift of length $\Delta=60$ performs best for $\eps < 4\times 10^{-3}$ whereas
the shift of length $\Delta=52$ performs best for $\eps > 4\times 10^{-3}$. The relatively larger column-span of the former helps for small values of $\eps$ whereas the relatively larger column distance of the latter helps for lager values of $\eps$. For $\eps \approx 3\times 10^{-3}$,  the RLC achieve a loss-probability of $\approx 4\times 10^{-5},$ the SCo codes achieve $\approx 6 \times 10^{-5}$ whereas the proposed constructions achieve $\approx 6 \times 10^{-6}$. More generally over the entire range of $\eps,$ the best embedded-RLC codes achieve a loss rate which is a factor of $10$ or more smaller than the SCo code and between a factor of $3$ to over $10$ smaller than the RLC code.

\begin{figure*}
  \begin{minipage}[b]{0.5\linewidth}
    \centering
    \includegraphics[width=\linewidth]{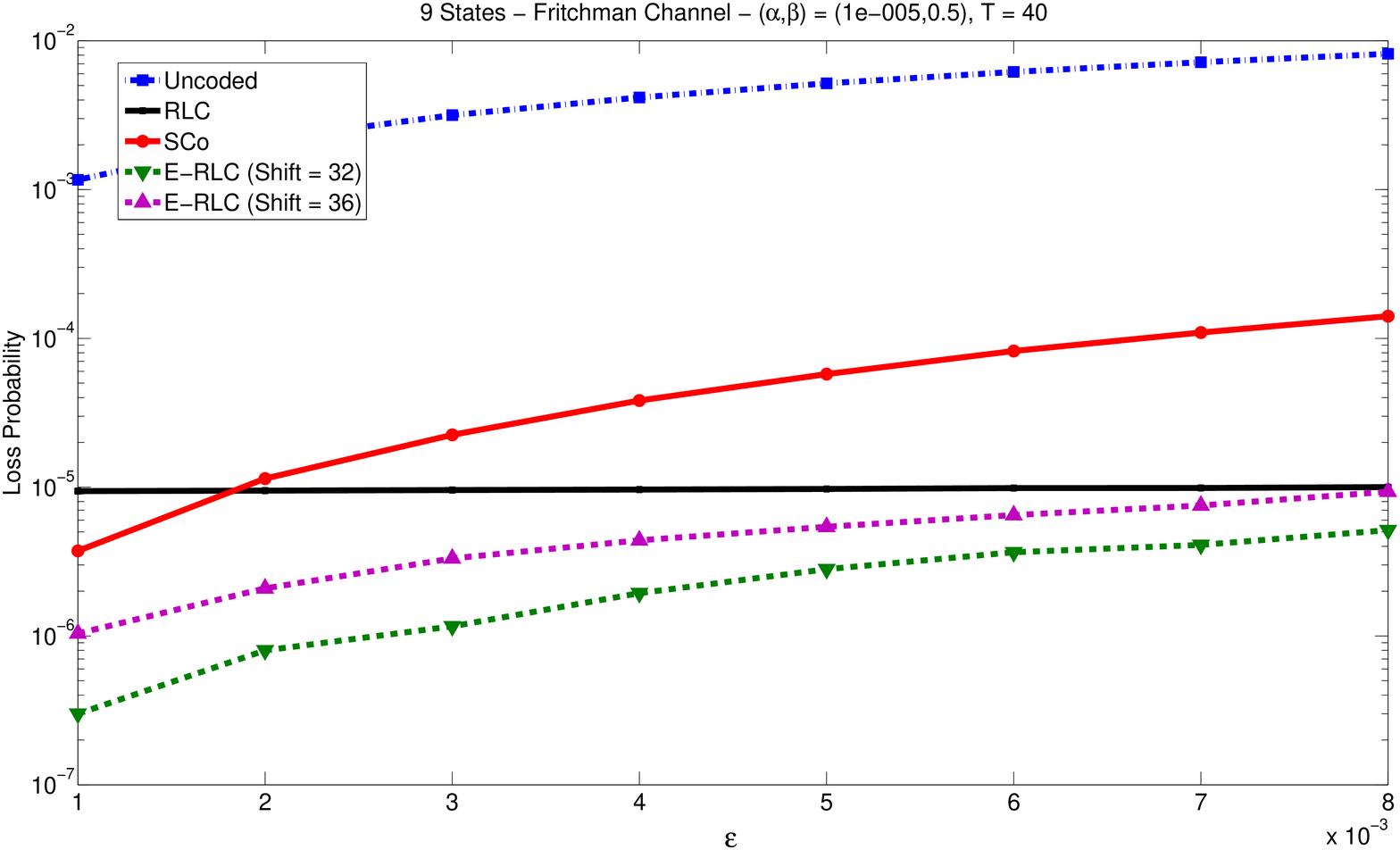}
    \caption{Simulation over a $N+1 = 9$-States Fritchman Channel with $(\al,\beta) = (10^{-5},0.5)$. All codes are evaluated using a decoding delay of $T=40$ symbols.}
    \label{fig:Fritchman_T40}
  \end{minipage}
  \hspace{0.5cm}
    \begin{minipage}[b]{0.5\linewidth}
    \centering
\includegraphics[width=\linewidth]{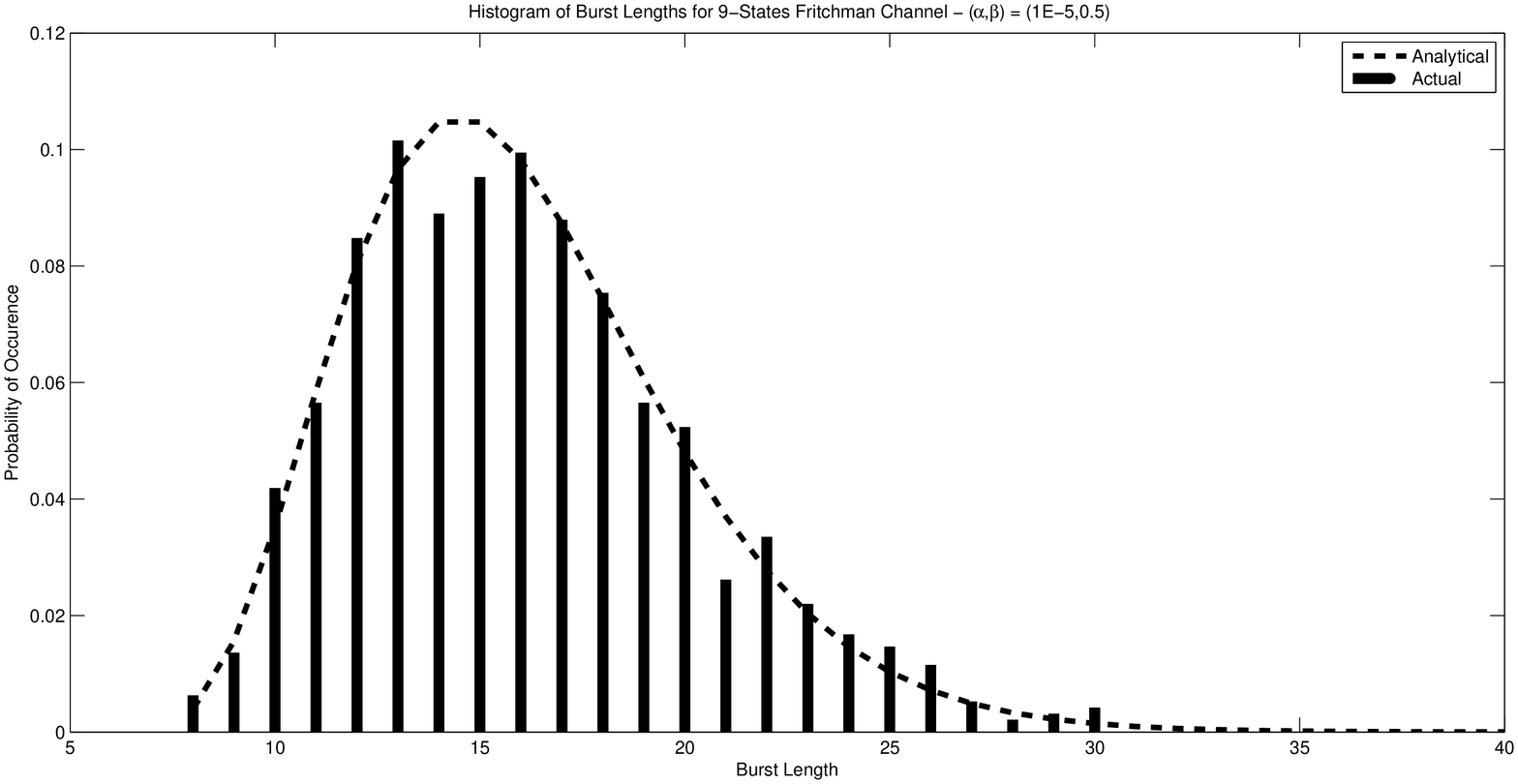}
\caption{Histogram of Bursts Lengths when $\beta =0.5$ in a $N+1 = 9$-States Fritchman Channel. The distribution follows a negative binomial distribution (shown dotted) of $N = 8$ failures and a success probability of 0.5.}
\label{fig:Fritchman_T40_Burst}
  \end{minipage}
\end{figure*}

\begin{figure*}
  \begin{minipage}[b]{0.5\linewidth}
    \centering
    \includegraphics[width=\linewidth]{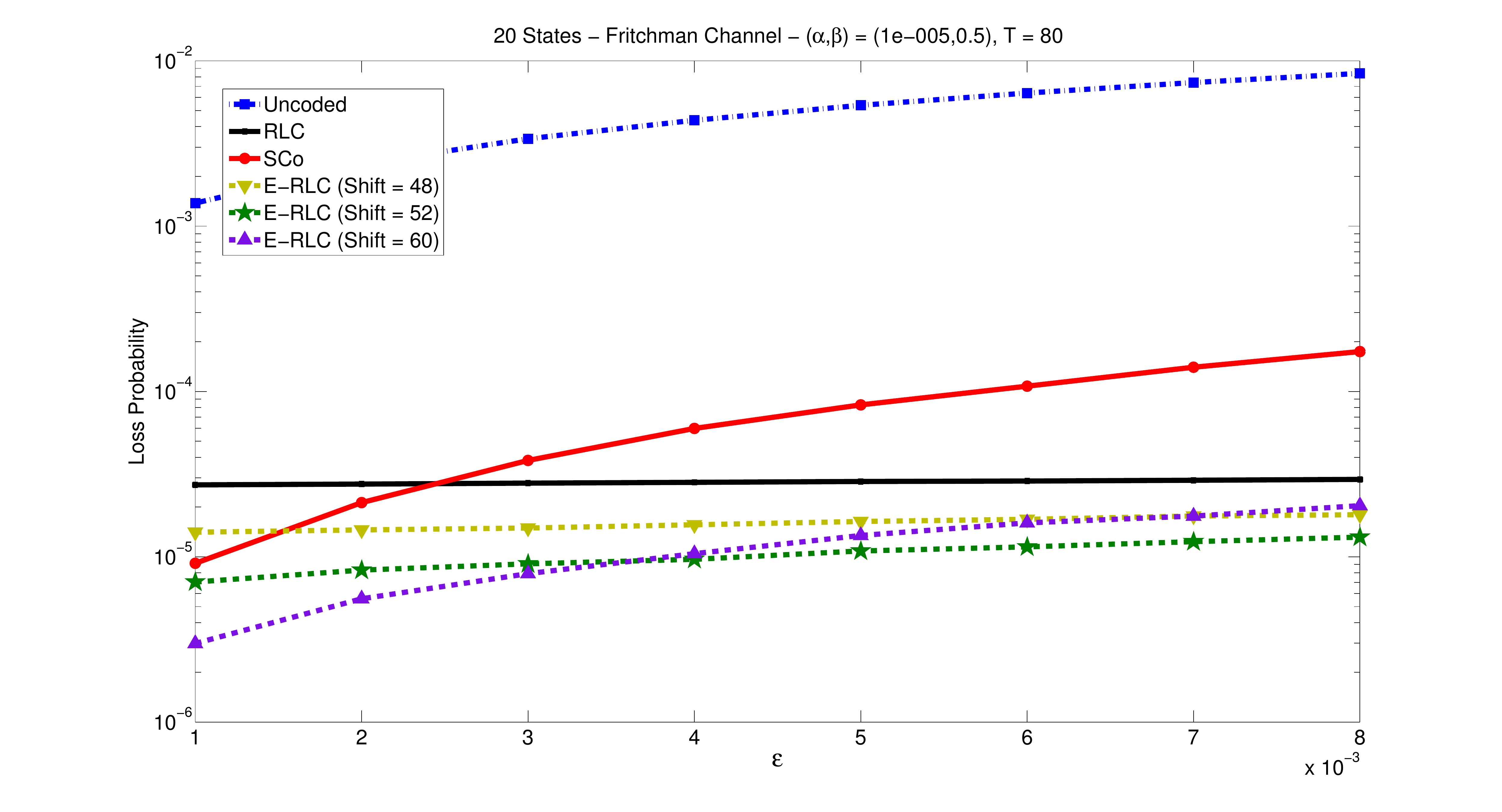}
    \caption{Simulation over a $N+1 = 20$-States Fritchman Channel with $(\al,\beta) = (10^{-5},0.45)$. All codes are evaluated using a decoding delay of $T=80$ symbols.}
    \label{fig:Fritchman_T80}
  \end{minipage}
  \hspace{0.5cm}
    \begin{minipage}[b]{0.5\linewidth}
    \centering
\includegraphics[width=\linewidth]{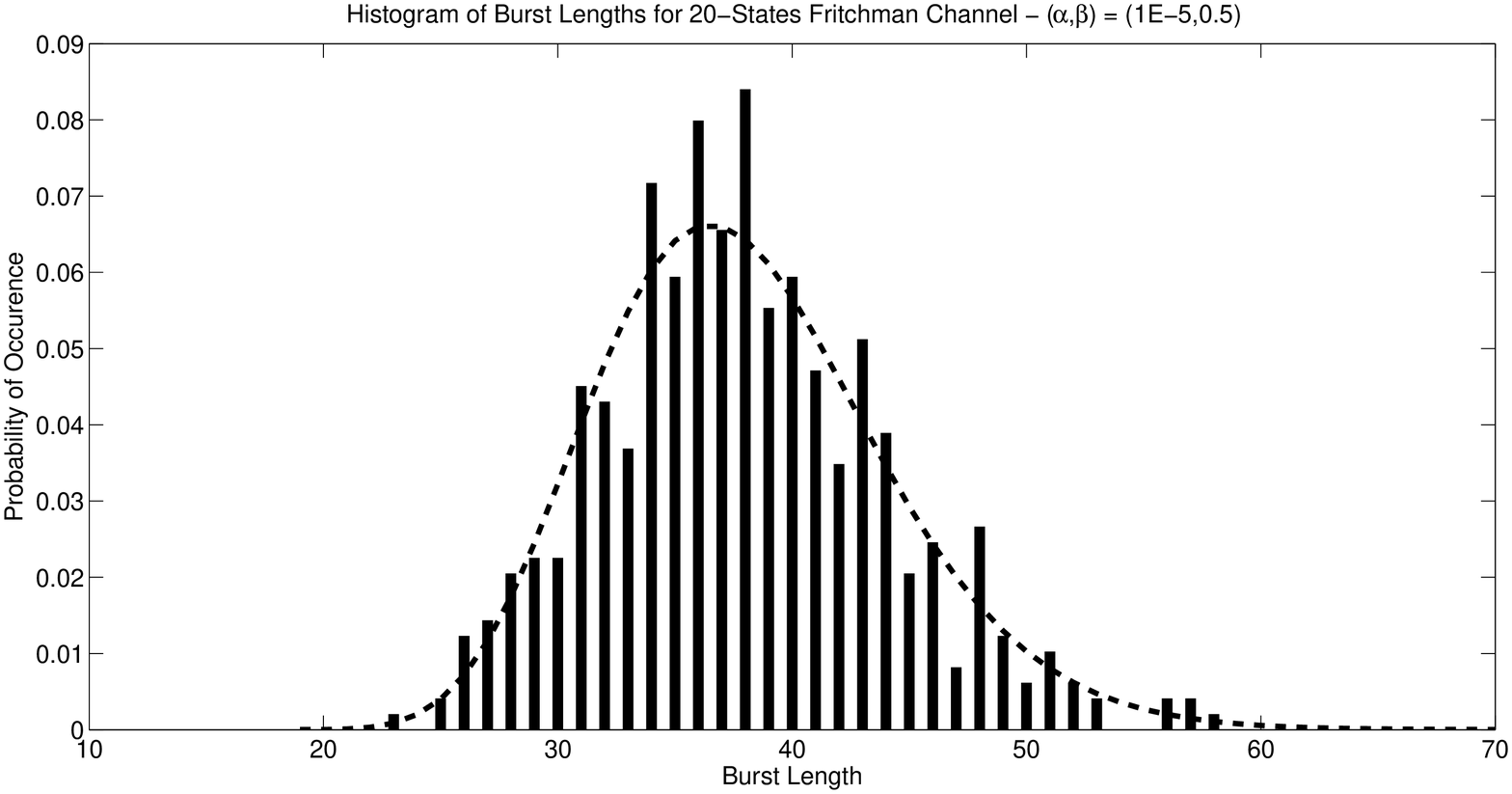}
\caption{Histogram of Bursts Lengths when $\beta =0.5$ in a $N+1 = 20$-States Fritchman Channel. The distribution follows a negative binomial distribution (shown dotted) of $N = 19$ failures and a success probability of 0.5.}
\label{fig:Fritchman_T80_Burst}
  \end{minipage}
\end{figure*}

\section{Conclusion}

We study the construction of low-delay codes for streaming data over channels that introduce both isolated and burst packet losses.
We show that good code constructions for such channels should simultaneously have large column span and column distance.
We establish, to our knowledge, the first outer bound on the achievable column-span and column-distance tradeoff for any convolutional code
of a given rate. This allows us to establish that some of the code constructions previously obtained from a computer search are indeed optimal.
We propose a new class of codes --- embedded-random linear codes --- that divide each source packet into two groups of symbols, perform unequal error protection
and combine the resulting parity checks with a suitable shift. We develop closed form expressions for the column distance and column span for these codes and
demonstrate how the code parameters can be tuned to obtain a flexible tradeoff between the column distance and column span.  Our proposed code constructions
achieve the outer bound for rate $R=1/2$ and also reduce to the known constructions such as the random linear codes and burst-erasure codes at the extreme points.
Numerical simulations on the Gilbert-Elliott channel and Fritchman channel indeed show significant performance gains over previously proposed constructions.

In terms of future work, it will be interesting to investigate optimal code constructions  for rates other than $R=0.5$. While our proposed construction in this paper splits each source-packet into 
two groups, it remains to be seen whether more groups are needed in general. It might also be interesting to see if the outer bound on column-distance and column-span tradeoff can be tightened for certain rate values.  
Extending these results to systems involving more than one communication link is also of great importance.
Finally experimental results over realistic packet loss traces will naturally provide a more realistic assessment of the performance gains from our delay-optimized code constructions.

\end{document}